\documentclass[12pt]{article}

\usepackage{multirow}
\usepackage{siunitx}
\usepackage[margin=1in]{geometry}
\usepackage{setspace}
\usepackage{placeins}
\usepackage{algpseudocode}
\usepackage{amsmath}    
\usepackage{amssymb}    
\usepackage{amsthm}     
\usepackage{mathtools}  
\usepackage{nccmath}    
\usepackage{bm}         

\usepackage{graphicx}
\usepackage{subcaption}

\usepackage{booktabs}   
\usepackage{tabularx}   
\usepackage{siunitx}    

\usepackage[linesnumbered,ruled,vlined]{algorithm2e}
\SetKwInput{KwInput}{Input} 
\SetKwInput{KwOutput}{Output} 

\usepackage[colorlinks,citecolor=blue,urlcolor=blue]{hyperref}

\usepackage{natbib}

\usepackage[normalem]{ulem} 

\newcommand{\Wcal}{\mathcal{W}}

\def\R{\mathbb{R}} 

\def\bq{\boldsymbol{q}}

\theoremstyle{definition}
\newtheorem{definition}{Definition}[section]

\theoremstyle{plain}
\newtheorem{theorem}{Theorem}[section]
\newtheorem{proposition}[theorem]{Proposition}

\def\spacingset#1{\renewcommand{\baselinestretch}%
{#1}\small\normalsize}

\title{Fast distance computation of multivariate distributions via nonparanormal transport}
\author{
Edward Shao\thanks{These authors contributed equally to this work.} $\, ^{1}$, Junyoung Park\footnotemark[1] $\, ^{1}$,\\
Naresh Punjabi$^{2}$, Hui Jiang$^{1}$, Irina Gaynanova$^{1}$\\[1ex]
$^{1}$Department of Biostatistics, University of Michigan\\
$^{2}$Miller School of Medicine, University of Miami
}
\date{}

\begin{document}
\maketitle
\begin{abstract}
With the increasing availability of data objects in the form of probability distributions, there is a growing need for statistical methods tailored to distributional data. Distance measures, especially the pairwise distance matrix between data objects, provide the foundation for a wide range of modern data analysis methods, such as clustering, multidimensional scaling, and distance-based regression, among others. The Wasserstein distance is commonly used with distributional data due to its compelling optimal transport property. However, while the Wasserstein distance can be efficiently computed for univariate distributions, its application to multivariate distributions is limited due to high computational costs. To address these scalability issues, we introduce the Nonparanormal Transport (NPT) metric, a closed-form distance based on the flexible nonparanormal distribution family for modeling skewed and non-Gaussian multivariate data. Simulation studies demonstrate that NPT maintains a high level of agreement with the Wasserstein distance, while being at least 2,460 times faster than its efficient variants when computing a 100-distribution pairwise distance matrix in both 2 and 5 dimensions. We illustrate the utility of NPT through a multidimensional scaling analysis of bivariate oxygen desaturation distributions of 723 individuals with sleep apnea in the Sleep Heart Health Study.
\end{abstract}

\noindent%
\section{Introduction}\label{sec-intro}

Distributional learning, which handles probability distributions as data objects, has gained interest due to increasingly available distributional data, such as from wearable devices, brain functional connectivity patterns, and demographic profiles across countries \citep{petersenModelingProbabilityDensity2022}. These distributions can be viewed as objects in a metric space, where the Wasserstein distance quantifies the minimal transport cost required to transform one distribution into another, thereby encoding the intrinsic geometry of the underlying domain of distributions \citep{villani2009optimal}. Owing to this geometric appeal, the Wasserstein distance has become widely used in modern statistical and machine learning applications, such as clustering, generative modeling, and imaging analysis \citep{zhuang2022wasserstein, verdinelliHybridWassersteinDistance2019, arjovsky2017wasserstein, kolouriOptimalMassTransport2017, kolouriSlicedWassersteinAutoEncoders2019}.

However, computing the Wasserstein distance for multivariate distributions poses significant challenges. While the Wasserstein distance admits closed-form expressions for univariate distributions and multivariate Gaussian distributions, such formulas are generally unavailable for more complex multivariate distributions. In those cases, the distance must be computed using optimization-based methods, which are computationally intensive. Given two empirical distributions with $n$ observations for each, their Wasserstein distance is typically computed through a linear programming approach, which requires $O(n^3\log n)$ computations \citep{cuturi2013sinkhorn}, becoming quickly unscalable as $n$ grows.
This problem is exacerbated when constructing a pairwise distance matrix across $N$ distributions, a central object when applying methods such as clustering or multidimensional scaling. As Wasserstein distance needs to be computed independently each time, this leads to $O(N^2n^3\log n)$ computations.

To mitigate such high computational costs, several variants of the Wasserstein distance have been proposed. The Hybrid Wasserstein distance was developed for distributional clustering purposes, based on the closed form of the Wasserstein distance for the Gaussian distribution, with the remaining discrepancy captured by tangent space approximations \citep{verdinelliHybridWassersteinDistance2019}. The sliced Wasserstein distance compares distributions by performing $L$ random projections onto one-dimensional subspaces, where the Wasserstein distance admits a closed form, followed by averaging the projected distances \citep{bonnotte2013unidimensional, bonneelSlicedRadonWasserstein2015}. Alternatively, entropic regularization can be incorporated into the Wasserstein distance optimization, enabling fast computation through the Sinkhorn algorithm \citep{cuturi2013sinkhorn}. Although these variants are faster than the Wasserstein distance, they can still be computationally costly and require parameter tuning that affects both accuracy and runtime.

To overcome these limitations, we introduce the Nonparanormal Transport (NPT) metric, a closed-form distance between multivariate distributions. NPT is defined on the space of nonparanormal distributions, that is, those following the Gaussian copula model \citep{liuNonparanormalSemiparametricEstimation2009}, which go beyond multivariate Gaussian distributions by allowing flexible modeling of univariate marginals, including skewed distributions. The NPT metric consists of two distinct components: the univariate Wasserstein distances across all marginals and the Wasserstein distance measuring the discrepancy in the underlying latent Gaussian dependency structure. Both components are closed form, leading to fast computations.

A particular computational advantage of NPT emerges when calculating a pairwise distance matrix across $N$ distributions. The computation of the first component of NPT, univariate Wasserstein distances across all marginals, requires encoding each marginal distribution with a quantile function. The computation of the second component of NPT requires estimation of the correlation matrix for the latent Gaussian distribution based on Kendall's $\tau$ \citep{liuHighdimensionalSemiparametricGaussian2012}. Both quantile functions and correlation matrices can be computed once for each of $N$ distributions in $\R^d$, requiring $O(Nd^2n\log n)$ computations. Once these precomputations are done, the subsequent computation of the NPT distance matrix is very fast, with full computation complexity being dominated by the precomputation step. Simulation studies confirm this, showing that NPT can construct distance matrices for $N=100$ distributions at least 2,460 times faster than competing metrics for $d=2,5$, while maintaining comparable, and often superior, accuracy in approximating the true Wasserstein distance.

In summary, our main contribution is the introduction of a novel, closed-form NPT metric for multivariate distributions. NPT is significantly faster than existing alternatives, while maintaining high concordance with the true Wasserstein distance in simulations.  We illustrate the utility of NPT through a multidimensional scaling analysis of bivariate oxygen desaturation distributions of 723 individuals with sleep apnea in the Sleep Heart Health Study \citep{Punjabi2008}.

\section{Methodology}\label{sec-meth}

In this section, we introduce the NPT metric for computing pairwise distances between multivariate distributions at large scale.
In Section~\ref{sec:wasserstein-review}, we first review the 2-Wasserstein distance.
Inspired by its known closed-form for univariate and multivariate Gaussian distributions, we define the NPT metric under the nonparanormal model and describe the computational details in Section~\ref{sec:npt-def}. We also discuss standardized distance considerations in Section~\ref{sec:scale}.

We let $\{P_i\}_{i=1}^N$ be a sample of $N$ continuous multivariate distributions on $\R^d$ with finite second moments: $\int_{\R^d} \|x\|^2 dP_j(x) < \infty$, where $\|x\|$ denotes the $\ell^2$ Euclidean metric of $x\in\R^d$. We let $X^{(i)} = (X^{(i)}_{1}, \dots, X^{(i)}_d)^\top \in \mathbb{R}^d$ be a random vector with distribution according to $P_i$, denoted as $X^{(i)}\sim P_i$. For each $P_i$, we assume we observe $n_i$ realizations of $X^{(i)}$, denoted by $x^{(i)}_k =  (x_{k,j}^{(i)})_{j=1}^d \in \R^d$ for $k = 1, \dots, n_i$ and $j = 1, \dots, d$. We let $\mathbf{1_N}$ denote the vector of all ones of length $N$, and let $\Phi$ denote the cdf of the univariate standard normal distribution.

\subsection{Wasserstein distance}\label{sec:wasserstein-review}

 The 2-Wasserstein distance between distributions $P_i$ and $P_l$ is defined as   
\begin{equation}\label{eq:multivariate}
    d^2_\mathcal{W}(P_i, P_l) = \inf_{\gamma \in \Gamma(P_i,P_l)} \int_{\mathbb{R}^d \times \mathbb{R}^d} \|x - y\|_2^2 \, d\gamma(x, y),
 \end{equation}
 where the infimum is taken over the set $\Gamma(P_i,P_l)$ of all possible joint distributions on $\mathbb{R}^d \times \mathbb{R}^d$ having marginals $P_i$ and $P_l$. For brevity, we omit the index and refer to this as the Wasserstein distance throughout the manuscript. For random vectors $X^{(i)}\sim P_i$ and $X^{(l)}\sim P_l$, we also write $d^2_\Wcal (X^{(i)}, X^{(l)}) = d^2_\Wcal(P_i, P_l)$ interchangeably.  The Wasserstein distance~\eqref{eq:multivariate} is an attractive distance choice due to its optimal transport interpretation \citep{villani2009optimal}. However, its empirical computation is challenging in dimensions $d\ge 2$ and requires solving a linear program with complexity $O(n^3 \log n)$ where $n = \max(n_i, n_l)$ \citep{cuturi2013sinkhorn}.

There are two special cases where the computations are significantly reduced due to the closed-form of~\eqref{eq:multivariate}. 
 When $d = 1$, the Wasserstein distance can be computed using the quantile functions: letting $F_{i}$ denote the cumulative distribution function of $P_i$ and $q_{i}(p) = \inf\{x \in \mathbb{R}: F_{i}(x) \geq p\}$ the corresponding quantile, the infimum in  \eqref{eq:multivariate} simplifies to\begin{equation}\label{eq:univariate}
     d^2_\mathcal{W}(P_i,P_l) = \int_{0}^{1} [q_{i}(p) - q_{l}(p)]^2 dp.
 \end{equation}
 When $d\geq 1$ and underlying distributions are multivariate Gaussian, that is $P_i= \mathcal{N}_d(m_i,\Sigma_i)$ and $P_l=\mathcal{N}_d(m_l,\Sigma_l)$, \eqref{eq:multivariate} simplifies to
\begin{equation}\label{eq:normal}
d_\mathcal{W}^2(P_i, P_l) = \| m_i - m_l \|_2^2 + \mathrm{Tr}\Big[ \Sigma_i + \Sigma_l - 2 (\Sigma_i^{1/2} \Sigma_l \Sigma_i^{1/2})^{1/2} \Big].
\end{equation}
 The trace term corresponds to the squared Bures distance, a metric on the space of positive semi-definite matrices \citep[Remark~2.31]{peyreComputationalOptimalTransport2019}.

\subsection{Nonparanormal transport metric}\label{sec:npt-def}

While~\eqref{eq:normal} provides a convenient closed-form for multivariate Gaussian distributions, the assumption of Gaussianity is too restrictive. We propose to relax this assumption by considering the nonparanormal distribution family.

\begin{definition}
[Nonparanormal Distribution \citep{liuNonparanormalSemiparametricEstimation2009}]
    The continuous random vector $X^{(i)}$ is said to follow a \textit{nonparanormal distribution} if there exist monotone increasing functions $f_i = (f_{i,j}:\R\to\R)_{j=1}^d$ such that
 \[
 Z_{i,j} = f_{i,j}(X^{(i)}_{j}) \quad \text{and} \quad Z_i = (Z_{i,1},\ldots,Z_{i,d})^\top \sim \mathcal{N}_d(0, \Sigma_i),
 \] 
 where $\Sigma_i$ is the latent correlation matrix. In this case, we also write
$
 X^{(i)} \sim \operatorname{NPN}(0, \Sigma_i, f_i). 
$ 
\end{definition}

The nonparanormal distribution allows the random vector $X^{(i)}$ to have flexible marginal distributions, as the monotone transformation $f_{i, j} = \Phi^{-1}\circ F_{i, j}$, where $F_{i,j}$ is the cumulative distribution function of $X^{(i)}_{j}$ and $\Phi^{-1}$ sends $X^{(i)}_{j}$ to the standard Gaussian variable $Z_{i, j}$. The dependencies between such flexible marginals are captured by the latent correlation matrix $\Sigma_i$ of the latent Gaussian vector $Z_i$. 
Consequently, the joint nonparanormal distribution $P_i$ of $X^{(i)}$ is uniquely characterized by the univariate marginals of $X^{(i)}$ and the latent correlation matrix $\Sigma_i$.
 
 Motivated by this structure, we propose a surrogate Wasserstein distance within the nonparanormal model, which we call the nonparanormal transport metric or NPT.
 \begin{definition}[NPT]
     For $X^{(i)} \sim \mathrm{NPN}(0, \Sigma_i, f_i) = P_i$ and $X^{(l)} \sim \mathrm{NPN}(0, \Sigma_l, f_l) = P_l$, we define the \textit{nonparanormal transport} (NPT) metric as 
 \begin{equation}\label{eq:proxy}
d^2(P_i, P_l) = d^2_{\mathrm{NPT}}(X^{(i)}, X^{(l)}) 
= \sum_{j=1}^{d} d^2_{\mathcal{W}}(X^{(i)}_{j}, X^{(l)}_{j}) 
+ \mathrm{Tr}\Big[\Sigma_i + \Sigma_l - 2 (\Sigma_i^{1/2} \Sigma_l \Sigma_i^{1/2})^{1/2}\Big].
\end{equation}

 \end{definition}

 In other words, $d_{NPT}$  consists of two parts, both of which have closed-form: the Wasserstein distance between univariate distributions \eqref{eq:univariate} (capturing the marginals), and the Wasserstein distance between latent Gaussian distributions \eqref{eq:normal} (capturing the dependency structure). It satisfies the metric axioms, as formalized below. 
\begin{proposition}\label{prop:npt-metric}
    The distance $d_{NPT}$ defines a metric on the space of nonparanormal distributions with finite second moments.
\end{proposition}

The computation of the NPT metric involves two precomputation steps for each distribution: estimating the quantiles for marginal distributions and estimating the latent correlation matrices.

First, to utilize the quantile-based closed form \eqref{eq:univariate}, we discretize the quantile function by evaluating empirical quantiles at the $m$ equally spaced grid points $0 = u_1 < \ldots < u_m = 1$, as considered in \citet{petersenFrechetRegressionRandom2019, coulterFastVariableSelection2025}. For each $j$th marginal of the distribution $P_i$, we compute and store the corresponding quantile vector $\bq_{ij} = (q_{ij}(u_1), \ldots, q_{ij}(u_m))^\top \in\R^m$. This single step involves sorting samples $\{x^{(i)}_{k, j}\}_{k=1}^{n_i}$ to estimate the empirical quantiles, which has complexity $O(n_i \log n_i)$,  followed by quantile evaluation at the grid points with complexity $O(m)$. Consequently, the total cost of this step for each distribution $P_i$ is $O\left(d(n_i \log n_i + m)\right)$.

Second, each latent correlation matrix $\Sigma_i$ is estimated based on Kendall's $\tau$ \citep{liuHighdimensionalSemiparametricGaussian2012}. For each pair of variables $j$ and $j'$, the Kendall's $\tau$ is computed as
\begin{equation}
\hat{\tau}_{j,j'}^{(i)} = \frac{2}{n_i (n_i - 1)} 
\sum_{1 \le k < k' \le n_i} 
\operatorname{sign}\big(x_{k,j}^{(i)} - x_{k',j}^{(i)}\big)
\operatorname{sign}\big(x_{k,j'}^{(i)} - x_{k',j'}^{(i)}\big),
\end{equation}
which is transformed to the latent correlation via
$$\hat\rho_{j, j'}^{(i)} = \sin\big(\pi \hat\tau_{j, j'}^{(i)} / 2\big).$$
This latent correlation computation is dominated by the cost of computing Kendall's $\tau$, which is $O(n_i\log n_i)$ per pair of coordinates, leading to overall cost $O(d^2n_i\log n_i)$ to form the elementwise estimator $\tilde\Sigma_i = (\hat\rho_{j, j'}^{(i)})_{j, j'=1}^d$. Because this elementwise estimator may not be positive semidefinite when $d>2$, we project each $\tilde\Sigma_i$ onto the positive definite cone, resulting in the final estimator $\widehat\Sigma_i$ with additional complexity $O(d^3)$. Thus, the total cost of estimating the latent correlation matrix $\Sigma_i$ is $O(d^2n_i\log n_i + d^3)$.

Given the quantile vectors $\{\bq_{ij}\}_{j=1}^d$ and latent correlation matrices $\widehat\Sigma_i$ for each $P_i$, we evaluate the marginal components of $d^2(P_i, P_l)$ in \eqref{eq:proxy} using the trapezoidal rule, and the latent Gaussian component using $\widehat\Sigma_i$ and $\widehat\Sigma_l$. The combined cost is $O(md + d^3)$. In practice, achieving a reliable estimation of each distribution $P_i$ typically requires $n_i \gg d$ (e.g., $n_i >d^2$ or $d^3$), and we set $m\le n_i$ since increasing $m$ beyond this yields negligible gains in approximating marginal quantiles. Therefore, the overall cost of NPT is largely dominated by the precomputation step, scaling roughly as $O(d^2n_i\log n_i)$ per distribution. This structure makes NPT particularly advantageous for distance matrix computations: with $N$ distributions and fixed $n=n_i$, $\forall i$, for simplicity, constructing the NPT distance matrix scales as $O(Nd^2n\log n + N^2(md +d^3))$. As the sample size $n$ is decoupled from the quadratic term $N^2$, this leads to faster computations than the Wasserstein distance ($O(N^2n^3\log n)$), the Sinkhorn distance ($O(N^2n^2)$ per iteration, with regularization parameter affecting the number of iterations required to reach convergence), and the Sliced Wasserstein distance with $L$ slices ($O(N^2Ln\log n)$).

\subsection{Standardized distance}\label{sec:scale}

The Wasserstein distance is known to be invariant to the fixed location shift, as for any $x\in \R^d$, it holds that
$
d^2_\Wcal (X^{(i)} + x, X^{(l)} + x) = d^2_\Wcal (X^{(i)}, X^{(l)})
$.
However, the Wasserstein distance is sensitive to scale. For univariate distance and scalar $a>0$, it holds that $d^2_{\mathcal{W}}(aX^{(i)}_{j}, aX^{(l)}_{j}) =  {\color{red} a^2} d^2_{\mathcal{W}}(X^{(i)}_{j}, X^{(l)}_{j})$ \citep{panaretos2019wasserstein}. Since the latent correlation component of NPT is scale-invariant, large scale differences in the marginal distributions can dominate the total distance, potentially obscuring meaningful differences in the underlying correlation structure. Furthermore, without adjustment, the dimensions corresponding to the largest scale will disproportionately dominate the marginal component.

To ensure that each marginal component contributes equally, we standardize the $N$ distributions across each $j$th dimension. Let $\mathcal{X}_j = \{ x_{k,j}^{(i)} : i=1,\dots,N;\; k=1,\dots,n_i \}$ be the pooled collection of realizations from the $j$th dimension across all $N$ distributions.  Let $\mathrm{med}(\mathcal{X}_j)$ be the corresponding median and $\mathrm{MAD}(\mathcal{X}_j)=\mathrm{med}_{x\in\mathcal{X}_j}\left|x-\mathrm{med}(\mathcal{X}_j)\right|$ the median absolute deviation. Then we scale realizations as
$
z_{k,j}^{(i)} = (x_{k,j}^{(i)} - \mathrm{med}(\mathcal{X}_j))/\mathrm{MAD}(\mathcal{X}_j). 
$
Subsequently, we evaluate NPT and Wasserstein distances using these standardized distributions.

\section{Simulation Study}\label{sec:sim}

We generate synthetic data to evaluate the runtime and accuracy of NPT relative to the Wasserstein distance and its existing variants. We consider two dimensions: bivariate distributions ($d=2$) and five-dimensional ($d=5$), with $N=100$ distributions generated for each setting. We fix the number of realizations at $n = n_i = 200$ for each of the $N$ distributions across all settings. Each $i$th distribution is defined as a nonparanormal distribution $P_i = \mathrm{NPN}\!\left(0,\, \Sigma_i,\, g_{\lambda_i}\right)$, where the latent correlation matrices follow a Toeplitz autocorrelation structure:
\[
\Sigma = \left( \rho_i^{\,|j-k|} \right)_{j,k=1}^{d}. 
\]
The marginal transformations are fixed across dimensions and defined by the mixture
\[
g^{-1}_{\lambda_i}(x) = (1-\lambda_i)e^{0.6x} + \lambda_i \bigl(0.8 \log(1 + 0.8 e^x)\bigr).
\]
The parameter pairs $(\rho_i,\lambda_i)$ are taken from a $10 \times 10$ rectangular grid with $\rho_i \in [-0.8, 0.8]$ and $\lambda_i \in [0,1]$ leading to $N=100$ distinct distributions. This grid-based construction ensures a comprehensive evaluation across a diverse range of distance values, encompassing distributions with varying degrees of marginal similarity and latent correlation. All distributions are standardized as in Section~\ref{sec:scale} prior to distance computation.

For the NPT distance, we employ $m=100$ discretization points for quantiles. For comparison, we use the Wasserstein distance, the Sliced Wasserstein distance with $L \in \{10, 100\}$ slices \citep{bonnotte2013unidimensional}, and the Sinkhorn distance with regularization $\varepsilon \in \{1, 10\}$ \citep{cuturi2013sinkhorn}, computed using the Python Optimal Transport library \citep{flamaryPOTPythonOptimal2021}.  We do not compare with Hybrid Wasserstein distance \citep{verdinelliHybridWassersteinDistance2019} as we were unable to find an open implementation. All distances are computed based on the $n = n_i = 200$ realizations from each of the $N$ distributions across all settings to ensure fair comparison of runtimes. Since Wasserstein distance based on $n=200$ realizations may not be an accurate approximation of the true Wasserstein distance in higher dimensions \citep{fournierRateConvergenceWasserstein2015}, we also compute Wasserstein distance based on 2000 realizations to be used as ground truth in accuracy comparisons.

To measure runtimes, we use the \texttt{microbenchmark} R package \citep{mersmann2024microbenchmark}. We evaluate both the single distance and the full pairwise distance matrix computation times. For the single-distance computation, we select two distributions at random from the pool of $N = 100$ distributions, and evaluate the mean time across 100 replications. For the full pairwise distance matrix, we consider only 5 replications. All evaluations are performed on a single CPU core (Intel Xeon Gold 6230) independently for each setting. 

To assess accuracy, we compare each method against the Wasserstein distance computed with 2000 realizations by constructing boxplots of the differences $ d^2_{\text{Wass}} - d^2_{\text{method}}$ across all $N(N-1)/2$ distribution pairs; here, we rescale the Sliced Wasserstein distance by a factor of $d$, thus using $d\cdot d_{\text{Sliced}}^2$, since it is known that $d \cdot d_{\text{Sliced}}^2 \le d_{\text{Wass}}^2$ when the number $L$ of slices tends to infinity \citep{bonnotte2013unidimensional}. We also evaluate concordance across pairs using Pearson and Spearman correlation coefficients.

\begin{table}[!t]
\centering
\spacingset{1}
\caption{Mean computation times (milliseconds) for single and matrix distance computations across dimensions in the simulation study. Standard deviations are shown in parentheses.}
\resizebox{0.9\columnwidth}{!}{%
\begin{tabular}{
l
S[table-format=4.0]
S[table-format=7.0]
S[table-format=4.0]
S[table-format=7.0]
}
\toprule
& \multicolumn{2}{c}{$d = 2$}
& \multicolumn{2}{c}{$d = 5$} \\
\cmidrule(lr){2-3} \cmidrule(lr){4-5}
\textbf{Method}
& \textbf{Single} & \textbf{Matrix}
& \textbf{Single} & \textbf{Matrix} \\
\midrule
NPT
& {\bfseries 1 (${<}1$)}
& {\bfseries 2 (${<}1$)}
& {\bfseries 2 (${<}1$)}
& {\bfseries 5 (${<}1$)} \\
Wasserstein (200 samples)
& {9 (${<}1$)}
& {42878 (529)}
& {9 (${<}1$)}
& {44461 (2006)} \\
Wasserstein (2000 samples)
& {991 (41)}
& {4931376 (26633)}
& {1202 (65)}
& {5773690 (13681)} \\
Sliced Wasserstein (10 slices)
& {3 (${<}1$)}
& {11919 (170)}
& {3 (${<}1$)}
& {12338 (473)} \\
Sliced Wasserstein (100 slices)
& {9 (${<}1$)}
& {44095 (124)}
& {9 (${<}1$)}
& {46117 (804)} \\
Sinkhorn ($\epsilon = 1$)
& {30 (26)}
& {81534 (6307)}
& {55 (30)}
& {150546 (8997)} \\
Sinkhorn ($\epsilon = 10$)
& {5 (5)}
& {16898 (54)}
& {6 (8)}
& {16969 (2037)} \\
\bottomrule
\end{tabular}%
}
\label{tab:benchmark_all_dimensions}
\end{table}

Table \ref{tab:benchmark_all_dimensions} presents the mean computation times (in milliseconds) for a single distance and the full distance matrix across all distances separately for each dimension. NPT is the fastest distance, with its computational advantage over other distances being particularly pronounced in matrix settings. For instance, when $d=5$, NPT is 8,890 folds faster than Wasserstein distance (200 samples) and 2,460 folds faster than Sliced Wasserstein distance (10 slices).

\begin{figure}[!t]
  \centering
  \spacingset{1}
  \includegraphics[width=0.9\textwidth]{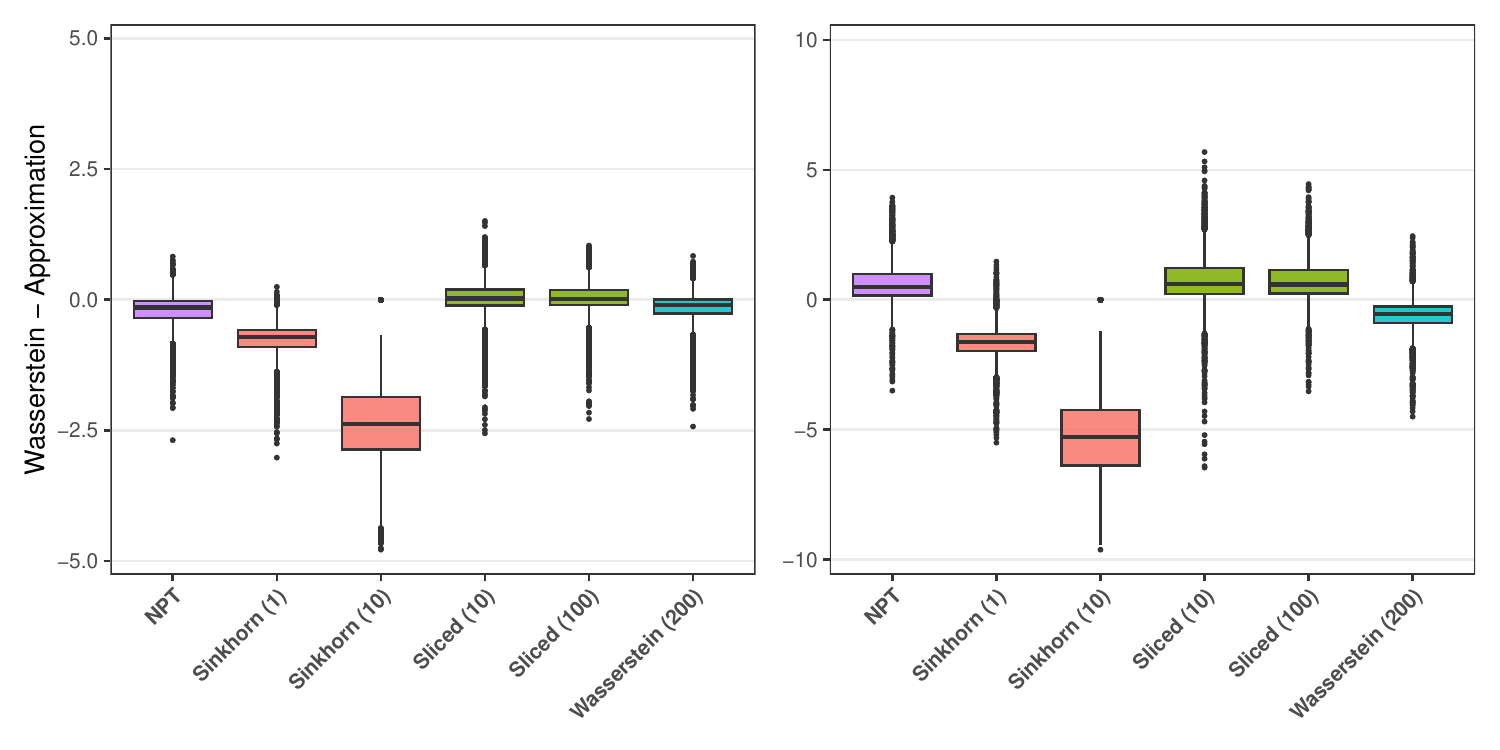}
  \caption{Absolute error comparison for $d=2$ (left) and $d=5$ (right). Differences $d^2_{\text{Wass}} - d^2_{\text{method}}$ are evaluated across $N(N-1)/2$ pairwise distances ($N=100$) with all methods based on $n=200$ realizations compared against the Wasserstein distance on 2000 realizations (treated as ground truth).}
  \label{fig:sim_error}
\end{figure}

Figure \ref{fig:sim_error} displays the boxplots of the absolute differences for each metric and dimension. The Sinkhorn distance exhibits the highest bias throughout. A small regularization parameter ($\epsilon = 1$) yields reduced bias and lower variability at the expense of higher computational cost, even slower than Wasserstein distance (200 samples); conversely, a larger value ($\epsilon = 10$) ensures faster computation but at the cost of significant bias and higher variability relative to the ground truth. Similarly, the Sliced Wasserstein distance demonstrates a trade-off: using 10 slices offers faster computation but higher variability, while 100 slices improves precision at a higher computational cost. Across both dimensions, NPT exhibits variance comparable to most other distances
and bias comparable to Sliced Wasserstein and Wasserstein distance ($n=200$). 
Notably, NPT achieves accuracy comparable to Wasserstein distance ($n=200$) in both dimensions. While Wasserstein distance ($n=200$) deviates further from the true distance as dimensionality increases from $d=2$ to $d=5$, which reflects the curse of dimensionality \citep{fournierRateConvergenceWasserstein2015}, NPT maintains a consistent error magnitude across dimensions, demonstrating superior robustness to increasing dimensionality. 
Complementary scatterplots and correlation coefficients are provided in Figures S1 and S2 in the Supplementary Material. At $d=2$, NPT achieves the highest alignment with the ground truth (Spearman $\rho = 0.93$), surpassing even the Wasserstein distance calculated at $n=200$ ($\rho = 0.92$). At $d=5$, NPT maintains a strong correlation of $\rho = 0.92$, only slightly exceeded by Sinkhorn ($\epsilon=1$) and the $n=200$ Wasserstein distance ($\rho = 0.96$). Overall, these results demonstrate that NPT closely approximates the underlying Wasserstein distance and preserves the relative ordering of distributional distances effectively, making it well-suited for downstream analyses.

\section{Application to SHHS data}\label{sec:data}

In this section, we analyze data from the Sleep Heart Health Study (SHHS) \citep{Punjabi2009}, a multicenter prospective cohort study designed to evaluate the consequences of sleep-disordered breathing (SDB). SHHS enrolled participants aged 40 years or older between 1995 and 1998 who had no prior treatment of SDB.  Participants underwent overnight home polysomnography using a portable monitor. The ABOSA software \citep{karhu2022abosa} was used to identify overnight desaturation events (temporary drops in oxygen saturation occurring during non-wake sleep stages), and to extract the depth and duration of each event (at least 4\% deep and at least 7 seconds long).

Traditionally, the oxygen desaturation index (ODI), defined as the number of desaturation events divided by the total sleep duration (without wake periods), serves as the primary metric of SDB severity. The ODI $\ge 30$ threshold represents severe sleep apnea, and is significantly associated with increased all-cause and cardiovascular mortality \citep{Punjabi2008}. While ODI provides a summary count of events, it collapses the underlying heterogeneity in event characteristics. We focus our analysis on the $n=723$ subjects meeting the severe SDB criteria (ODI $\ge 30$). In this cohort, each subject's empirical bivariate distribution (depth and duration) is based on a substantial number of desaturation events (averaging 335 events per subject, range 136–962). Figure~\ref{fig:dist_scatter} displays scatterplots with kernel density estimates for three randomly selected subjects spanning lower, moderate, and higher ODI levels, illustrating substantial heterogeneity in the joint distribution of depth and duration.

On this resulting distributional data, we pursue three analysis goals. For all analyses, we standardize each subject-level distribution as described in Section~\ref{sec:scale}. First, we compare the computational performance of the squared NPT metric relative to the squared 2-Wasserstein distance and its variants, mirroring the settings of our simulation studies. Second, we compare accuracy. Third, we construct multidimensional scaling embeddings from the pairwise squared-distance matrix of NPT and compare the resulting visualizations with those obtained from the squared Wasserstein distance. We also examine whether these representations align with clinical characteristics, including ODI and age. 

\begin{figure}[!t]
 \centering
 \spacingset{1}
 \includegraphics[width=0.8\textwidth]{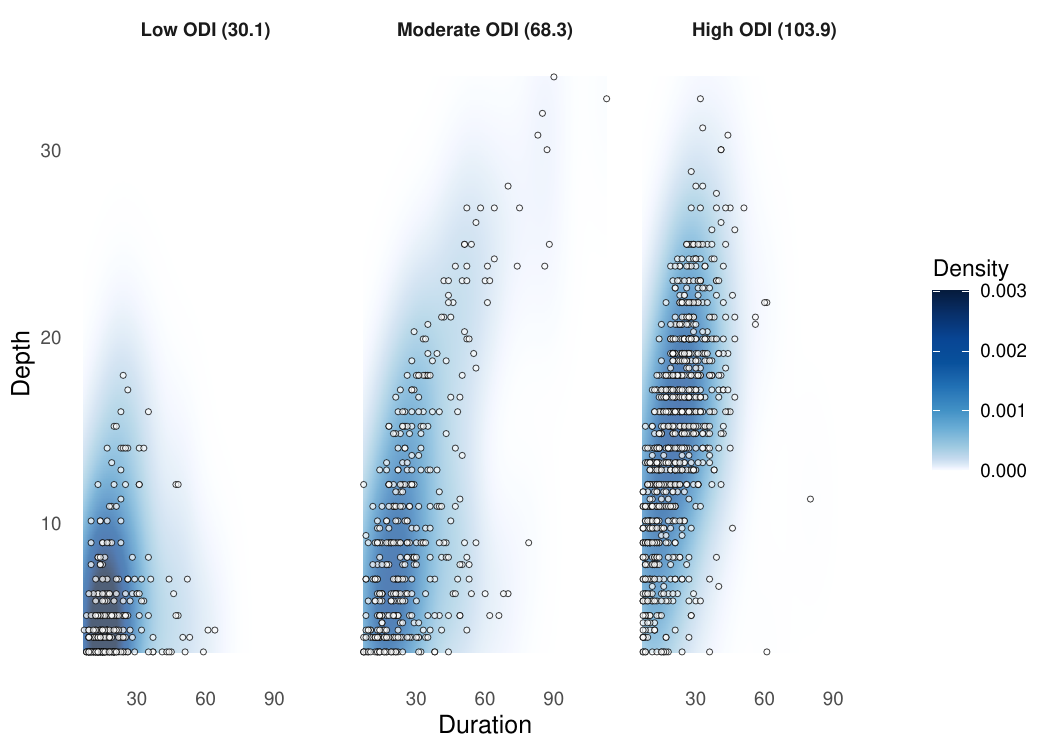}
\caption{Scatter plots with overlaid kernel density estimates of bivariate measurements (duration and depth) of desaturation events for three SHHS subjects. Values in parentheses denote the corresponding ODI for each subject.}
 \label{fig:dist_scatter}
\end{figure}

\paragraph{Run Time.}
We record runtimes for computing both a single pairwise distance and the full distance matrix for all 723 subjects using each squared distance metric, measured with the \texttt{microbenchmark} package \citep{mersmann2024microbenchmark}. The single distance is computed between the two subjects with the most desaturation events (962 and 919) over 100 replications, while the full distance matrix is computed over 5 replications. The runtimes were recorded using a single CPU core (Intel Xeon Gold 6230).

Table~\ref{tab:benchmark_time} reports the average runtime for each distance. NPT significantly outperforms the Wasserstein distance and its variants, achieving at least a 3$\times$ speedup for single distance computations and over a 1400$\times$ speedup for full distance matrix computations compared to the fastest variant (Sliced Wasserstein with 10 slices).

\begin{table}[!t]
\centering
\spacingset{1}
\caption{Mean computation times (milliseconds) for single and matrix distance computations ($d=2$, $N=723$) on the SHHS data. Standard deviations are shown in parentheses.}
\resizebox{0.8\columnwidth}{!}{%
\begin{tabular}{
l
S[table-format=3.0(2)]
S[table-format=8.0(6)]
}
\hline
\textbf{Method} & 
{\textbf{Single Distance}} & 
{\textbf{Distance Matrix}} \\
\hline
NPT & \multicolumn{1}{c}{\bfseries 2 ($<1$)} & \bfseries 423(29) \\
Wasserstein & 246(28) & 6857471(220481) \\
Sliced Wasserstein (10 slices) & 6(3)& 622782(26464) \\
Sliced Wasserstein (100 slices) & 43(2) & 3738090(176530) \\
Sinkhorn ($\epsilon=1$) &  53(5)& 4375114(302144) \\
Sinkhorn ($\epsilon=10$) & 38(5) & 1538067(1835) \\
\hline
\end{tabular}%
}
\label{tab:benchmark_time}
\end{table}

\paragraph{Metric agreement.} 

\begin{figure}[!t]
 \centering
 \spacingset{1}
 \includegraphics[width=1.0\textwidth]{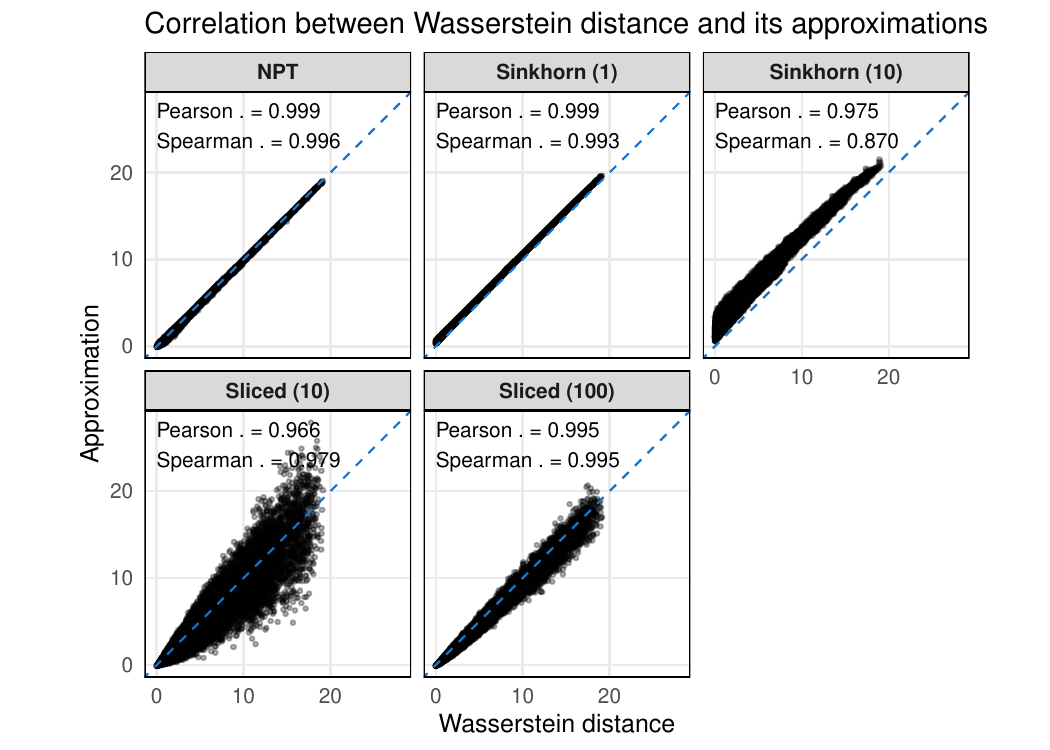}
 \caption{Scatterplot comparing distances to the reference Wasserstein distance with correlation coefficients for SHHS data. The horizontal axis represents the Wasserstein distance, while the vertical axis displays the estimated distance for each approximation. Points falling on the dashed diagonal identity line ($y=x$) indicate a perfect match between the approximation and the Wasserstein distance.}
 \label{fig:scatter_real}
\end{figure}

Figure~\ref{fig:scatter_real} evaluates the agreement of NPT and Wasserstein-based variants with the Wasserstein distance, where the Sliced Wasserstein distance is multiplied by $d$ as in simulations. The scatterplots represent paired distances across all 723 subjects, with Spearman and Pearson correlation coefficients summarizing the strength of agreement. NPT shows the closest alignment with the Wasserstein distance, exhibiting minimal variance and no visual bias, with both correlation coefficients very close to 1. 

Figure S3 in the supplementary material presents boxplots of the pairwise differences 
$d^2_{\text{Wass}} - d^2_{\text{method}} $ comparing NPT and the Wasserstein-based variants to the Wasserstein distance across all 723 participants. Consistent with the simulation results, NPT exhibits variance and bias comparable to the variants of Wasserstein distance, while producing fewer outliers.

\paragraph{Multidimensional Scaling (MDS).} 
MDS is a dimensionality reduction technique that constructs a low-dimensional Euclidean embedding of the data based on the pairwise distance matrix \citep{GroenenVelden2014MDS}. Specifically, for $N$ distributions, given $N\times N$ matrix $D^{(2)}$ of squared pairwise distances, MDS produces embeddings based on the spectral decomposition of the matrix
\[
B = -\tfrac12\, J D^{(2)} J, \qquad
J = I_N - \tfrac1N \mathbf{1_N}\mathbf{1_N}^\top.
\]
Here, $J$ is the centering matrix that ensures the resulting embeddings are centered at the origin, and $B$ represents the approximated inner products of the points in a latent Euclidean space. The coordinates of the new $r$-dimensional embedding, with $r < N$, are derived from the top $r$ eigenvalues and their corresponding eigenvectors, where we fix $r=2$ for illustration. Thus, MDS effectively transforms the distances into a 2D scatterplot, where the proximity between points serves as a direct visual proxy for the similarity between their original distributions.

We present visualizations from MDS applied to the squared NPT and Wasserstein distance matrices for all 723 subjects.
\begin{figure}[!t]
 \centering
 \spacingset{1}
 \includegraphics[width=1\textwidth]{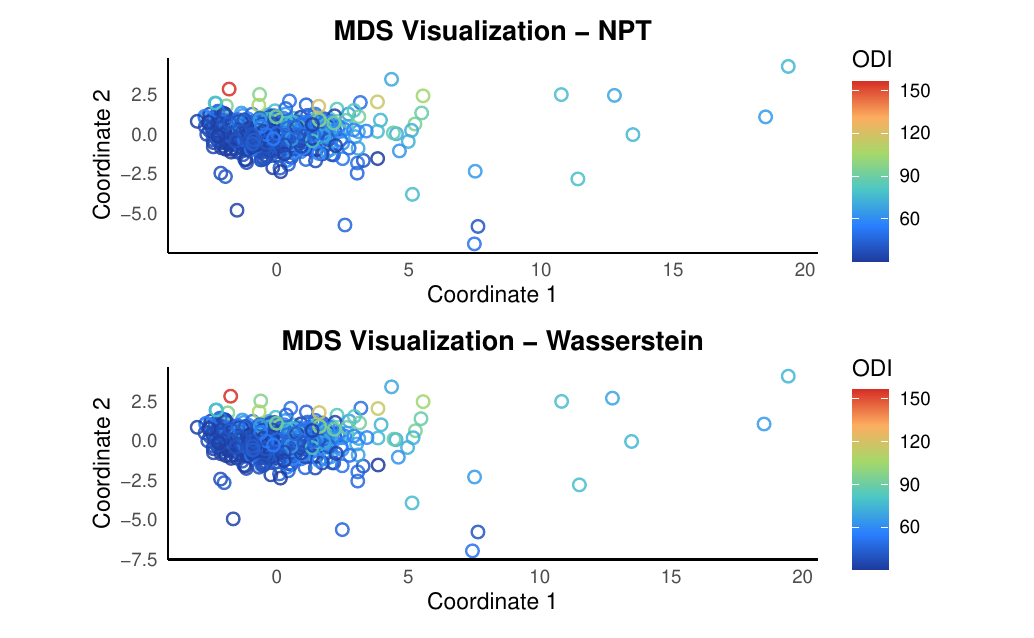}
 \caption{Scatterplots of the first two coordinates from a 2-dimensional MDS for NPT (top) and Wasserstein (bottom) distance metrics on SHHS data. Hollow circles represent individual distributions and are color-coded by ODI.}
 \label{fig:mds_odi}
\end{figure}
Figure \ref{fig:mds_odi} provides a comparison of scatterplots of the two coordinates from a 2-dimensional MDS of the squared NPT and Wasserstein distances, with points color-coded by ODI score. The figure shows that the MDS embeddings derived from NPT align closely with those derived from the Wasserstein distance. 
Additionally, for both NPT and Wasserstein distances, higher positive values along the second MDS coordinate tend to correspond to higher ODI values, revealing an ODI-related structure in the embedding.
We also analyze MDS visualizations derived from NPT and the Wasserstein distance using age (see Figure S4 in the Supplementary Material). The MDS visualizations reveal that lower values along the second coordinate tend to correspond to older individuals, capturing an age-related pattern in the embedding.

\section{Discussion}\label{sec-disc}
In this work, we propose a novel distance metric for quantifying the distance between distributional objects, termed the nonparanormal transport (NPT) metric. As a computationally efficient alternative to the Wasserstein distance, the NPT metric provides a scalable framework for calculating distances between multivariate distributions and constructing the resulting pairwise distance matrix. Our numerical studies, encompassing both simulated scenarios and the SHHS dataset, demonstrate that the NPT metric closely aligns with the 2-Wasserstein distance while offering substantially improved computational performance. These results demonstrate that NPT effectively removes the computational bottleneck associated with solving the optimal transport problem for the Wasserstein distance, enabling such multivariate analyses on a much larger scale than previously feasible.
Promising future directions for applying NPT include generative models, where Wasserstein distance and its variants have been instrumental in improving training stability, yet often at considerable computational costs \citep{arjovsky2017wasserstein, deshpande2018generative, kolouriSlicedWassersteinAutoEncoders2019}.

While NPT relies on a semiparametric family that is significantly more flexible than the Gaussian model, it is still more constrained than fully general multivariate distributions. In particular, nonparanormal models are ill-suited for capturing tail dependencies, such as those found in financial data, and they only accommodate continuous distributions, whereas binary, ordinal, and zero-inflated data are common in practice. However, extended distribution families based on latent Gaussian copulas have been developed to accommodate such data types \citep{fanHighDimensionalSemiparametric2017, huangLatentcorPackageEstimating2021, yoonSparseSemiparametricCanonical2020}, suggesting that NPT can potentially be extended to these domains in future work. Finally, while here we demonstrated empirically that NPT aligns closely with the Wasserstein distance, complementing these observations with theoretical analysis remains challenging, since closed-form expressions for the Wasserstein distance are generally unavailable even for the nonparanormal family, and thus constitutes an important avenue for future research. The accompanying R package \texttt{NonparaTransport} is available at \url{https://github.com/IrinaStatsLab/NonparaTransport}.

\section*{Funding}\label{Funds}
This research was supported by NIH R01HL172785.

\appendix

\section{Proof of Proposition~2.1}
\label{app:proofs-prop123}

In this section, we present the detailed proof of
Proposition~2.1.

\begin{proof}[Proof of Proposition~2.1]

Let
\[
  P_i=\mathrm{NPN}(0,\Sigma_i,f_i), \quad
  P_l=\mathrm{NPN}(0,\Sigma_l,f_l), \quad
  P_\alpha=\mathrm{NPN}(0,\Sigma_\alpha,f_\alpha)
\]
be any three nonparanormal distributions.
Recall that
\[
  d_{\mathrm{NPT}}(P_i,P_l)
  =
  \Bigg(
    \sum_{j=1}^{d}
    d_{\mathcal W}^2\!\bigl(X^{(i)}_{j},X^{(l)}_{j}\bigr)
    +
    d_{\mathcal B}^2(\Sigma_i,\Sigma_l)
  \Bigg)^{1/2}
\]
satisfies the metric axioms on the space of nonparanormal distributions
with finite second moments, where
\[
  d_{\mathcal B}(\Sigma_i,\Sigma_l)
  =
  \Bigl(
    \operatorname{Tr}\bigl[
      \Sigma_i + \Sigma_l
      - 2(\Sigma_i^{1/2}\Sigma_l\Sigma_i^{1/2})^{1/2}
    \bigr]
  \Bigr)^{1/2}.
\]
This quantity is the Bures distance between latent correlation matrices,
which satisfies the metric axioms on the set of positive semidefinite
matrices \citep{bhatiaBuresWassersteinDistance2019}.

To prove $d_{\mathrm{NPT}}$ is a metric on the space of nonparanormal
distributions, we verify the following axioms:
\begin{enumerate}
  \item \textit{Positivity:} $d_{\mathrm{NPT}}(P_i, P_l) \ge 0$.
  \item \textit{Identity:} $d_{\mathrm{NPT}}(P_i,P_l)=0 \iff P_i = P_l$.
  \item \textit{Symmetry:} $d_{\mathrm{NPT}}(P_i, P_l) = d_{\mathrm{NPT}}(P_l, P_i)$.
  \item \textit{Triangle inequality:}
        $d_{\mathrm{NPT}}(P_i,P_l) \leq
         d_{\mathrm{NPT}}(P_i,P_\alpha) + d_{\mathrm{NPT}}(P_\alpha,P_l)$.
\end{enumerate}
Positivity and symmetry follow immediately from the definition.

\paragraph{Identity.}

Suppose $P_i=P_l$. Then they share the same marginals
$X_j^{(i)} = X_j^{(l)}$ and the same latent correlation matrices
$\Sigma_i = \Sigma_l$, since a nonparanormal distribution is uniquely
determined by its marginal distributions and latent correlation matrix
\citep{liuNonparanormalSemiparametricEstimation2009}. By definition of $d_{\mathrm{NPT}}$, we
obtain $d_{\mathrm{NPT}}(P_i, P_l) = 0$.

Conversely, suppose $d_{\mathrm{NPT}}(P_i,P_l)=0$. Since each term in
$d_{\mathrm{NPT}}$ is nonnegative, it follows that
\[
  d_{\mathcal W}\!\bigl(X^{(i)}_{j},X^{(l)}_{j}\bigr)=0
  \quad \text{for all } j=1,\dots,d,
  \qquad
  d_{\mathcal B}(\Sigma_i,\Sigma_l)=0.
\]
Since $d_{\mathcal W}$ and $d_{\mathcal{B}}$ are metrics,
\[
  X^{(i)}_{j} = X^{(l)}_{j}
  \quad \text{for all}\quad j=1,\dots,d, \qquad\text{and}\qquad
  \Sigma_i = \Sigma_l.
\]
We thus conclude $P_i = P_l$, as a nonparanormal distribution is
uniquely determined by its marginal distributions and its latent
correlation matrix.

\paragraph{Triangle inequality.}

We consider the individual terms in $d_{\mathrm{NPT}}(P_i,P_l)$.
For each $j=1,\dots,d$, the univariate Wasserstein distance is a metric
on the space of univariate distributions and therefore satisfies
\[
  d_{\mathcal W}(X^{(i)}_{j},X^{(l)}_{j})
  \le
  d_{\mathcal W}(X^{(i)}_{j},X^{(\alpha)}_{j})
  + d_{\mathcal W}(X^{(\alpha)}_{j},X^{(l)}_{j}).
\]
Squaring both sides,
\[
  d^2_{\mathcal W}(X^{(i)}_{j},X^{(l)}_{j})
  \le
  \bigl(d_{\mathcal W}(X^{(i)}_{j},X^{(\alpha)}_{j})
       + d_{\mathcal W}(X^{(\alpha)}_{j},X^{(l)}_{j})\bigr)^2.
\]
On the other hand, the Bures distance is a metric on the space of
positive semi-definite matrices \citep{bhatiaBuresWassersteinDistance2019},
satisfying
\[
  d_{\mathcal B}(\Sigma_i,\Sigma_l)
  \le
  d_{\mathcal B}(\Sigma_i,\Sigma_\alpha)
  + d_{\mathcal B}(\Sigma_\alpha,\Sigma_l),
\]
and consequently
\[
  d^2_{\mathcal B}(\Sigma_i,\Sigma_l)
  \le
  \bigl(d_{\mathcal B}(\Sigma_i,\Sigma_\alpha)
       + d_{\mathcal B}(\Sigma_\alpha,\Sigma_l)\bigr)^2.
\]
Combining the individual squared inequalities and taking square roots,
\begin{align*}
  d_{\mathrm{NPT}}(P_i,P_l)
  &=
  \left(
    \sum_{j=1}^{d} d_{\mathcal W}^2(X^{(i)}_{j},X^{(l)}_{j})
    + d_{\mathcal B}^2(\Sigma_i,\Sigma_l)
  \right)^{1/2} \\
  &\le
  \left(
    \sum_{j=1}^{d}
      \Bigl( d_{\mathcal W}(X^{(i)}_{j},X^{(\alpha)}_{j})
           + d_{\mathcal W}(X^{(\alpha)}_{j},X^{(l)}_{j}) \Bigr)^2
    +
    \Bigl( d_{\mathcal B}(\Sigma_i,\Sigma_\alpha)
          + d_{\mathcal B}(\Sigma_\alpha,\Sigma_l) \Bigr)^2
  \right)^{1/2}.
\end{align*}
Since $d_{\mathcal W}, d_{\mathcal{B}} \geq 0$ and $d_{\mathrm{NPT}}$
is the $\ell^2$-norm sum of these metrics, we apply Minkowski's
inequality for sums \citep{rudinFunctionalAnalysis1991} to obtain
\begin{align*}
  d_{\mathrm{NPT}}(P_i,P_l)
  &\le
  \left(
    \sum_{j=1}^{d} d_{\mathcal W}^2(X^{(i)}_{j},X^{(\alpha)}_{j})
    + d_{\mathcal B}^2(\Sigma_i,\Sigma_\alpha)
  \right)^{1/2}
  +
  \left(
    \sum_{j=1}^{d} d_{\mathcal W}^2(X^{(\alpha)}_{j},X^{(l)}_{j})
    + d_{\mathcal B}^2(\Sigma_\alpha,\Sigma_l)
  \right)^{1/2} \\
  &= d_{\mathrm{NPT}}(P_i,P_\alpha) + d_{\mathrm{NPT}}(P_\alpha,P_l),
\end{align*}
which completes the proof.
\end{proof}

\FloatBarrier
\newpage

\section{Additional Simulation Results}

We present supplementary scatterplots from the simulation study,
comparing the alignment of NPT and variants of the Wasserstein distance
(plotted on the $y$-axis) against the reference Wasserstein distance
computed on 2000 samples ($x$-axis) across simulation settings with
$d = 2$ (Figure~\ref{fig:scatter_two}) and $d = 5$
(Figure~\ref{fig:scatter_five}) dimensions. The dotted blue lines
indicate perfect agreement.

\begin{figure}[!t]
  \centering
  \includegraphics[width=1.0\textwidth]{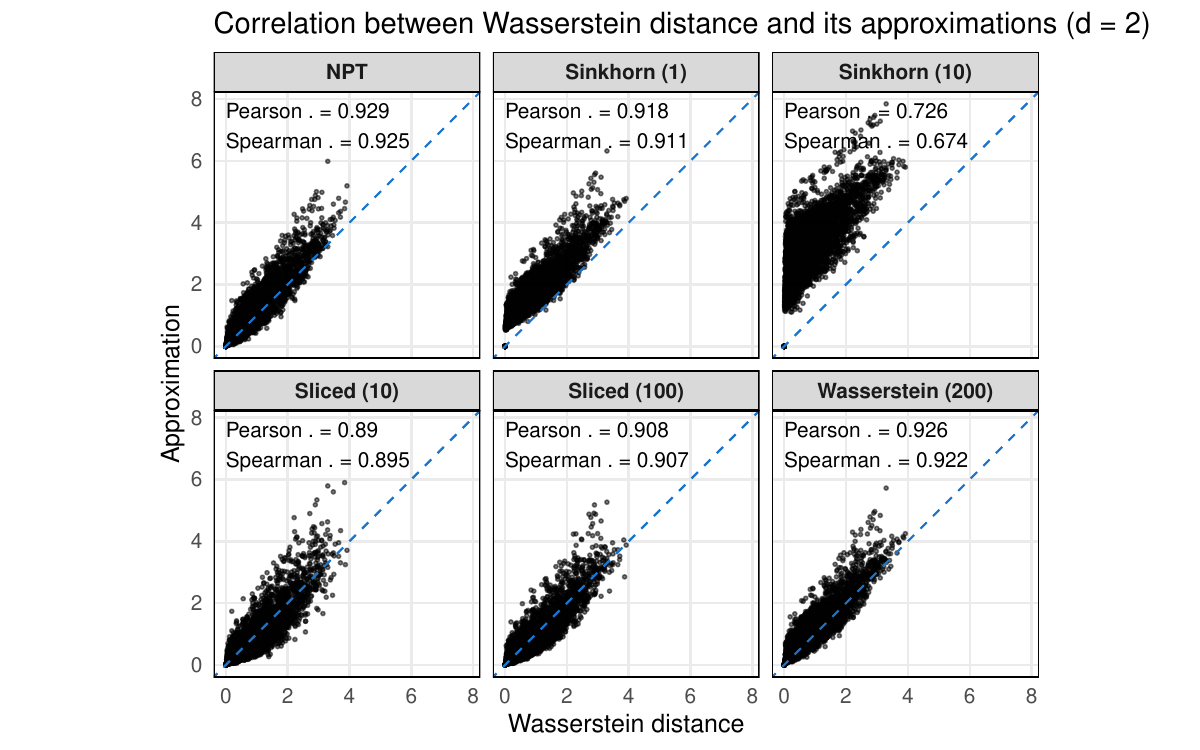}
  \caption{Scatterplot comparing distances to the reference Wasserstein
    distance (2000 samples) with correlation coefficients for the $d=2$
    simulation setting. The horizontal axis represents the Wasserstein
    distance estimated from 2000 samples for each distribution, while
    the vertical axis displays the estimated distance for each
    approximation. Points on the dashed diagonal identity line ($y=x$)
    indicate perfect agreement between the approximation and the
    reference.}
  \label{fig:scatter_two}
\end{figure}

\begin{figure}[!t]
  \centering
  \includegraphics[width=1.0\textwidth]{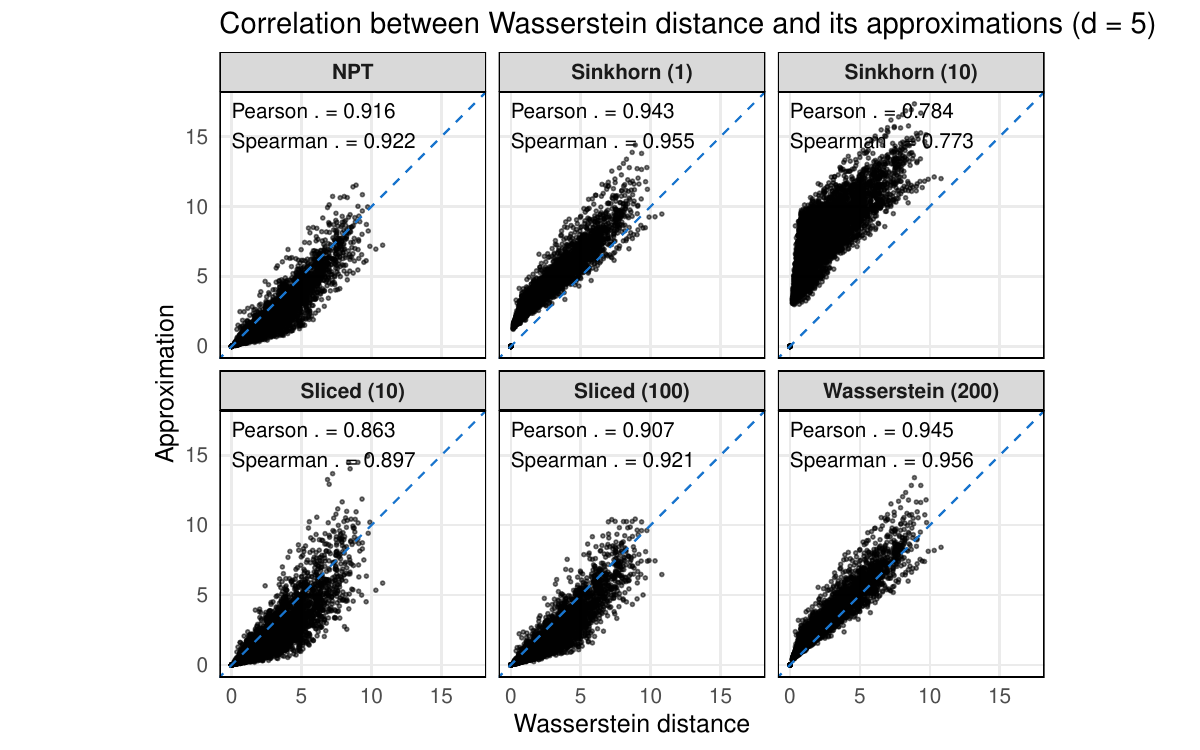}
  \caption{Scatterplot comparing distances to the reference Wasserstein
    distance (2000 samples) with correlation coefficients for the $d=5$
    simulation setting. The horizontal axis represents the Wasserstein
    distance estimated from 2000 samples for each distribution, while
    the vertical axis displays the estimated distance for each
    approximation. Points on the dashed diagonal identity line ($y=x$)
    indicate perfect agreement between the approximation and the
    reference.}
  \label{fig:scatter_five}
\end{figure}

\FloatBarrier
\newpage

\section{Additional Real Data Analysis Results}

We present additional visualizations from the real data analysis.
First, we show a boxplot of the absolute differences between the
Wasserstein distance and other distance measures for 723 SHHS subjects
with ODI $\geq 30$ (Figure~\ref{fig:real}). Next, we display
two-dimensional MDS plots of the NPT and Wasserstein distance matrices,
colored by age (Figure~\ref{fig:mds_age}).

\begin{figure}[!t]
  \centering
  \includegraphics[width=0.8\textwidth]{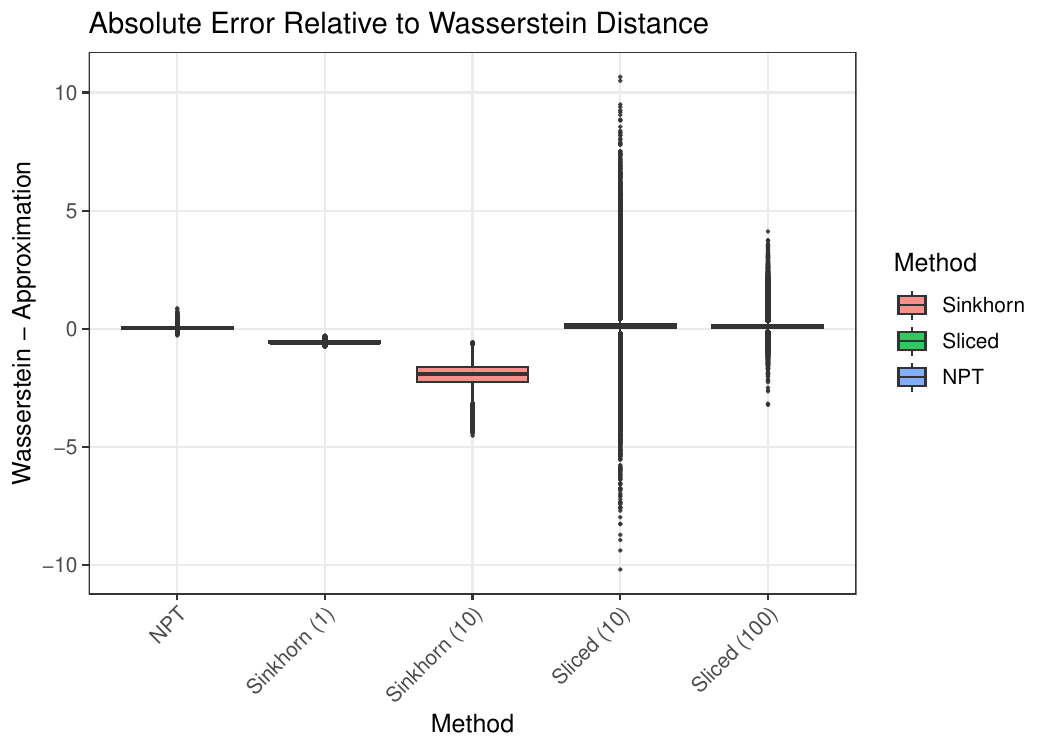}
  \caption{Absolute error comparison on SHHS data. Differences
    $d^2_{\text{Wass}} - d^2_{\text{method}}$ are evaluated across
    $N(N-1)/2$ pairwise distances ($N=723$) against the Wasserstein
    distance (treated as ground truth).}
  \label{fig:real}
\end{figure}

\begin{figure}[!t]
  \centering
  \includegraphics[width=1\textwidth]{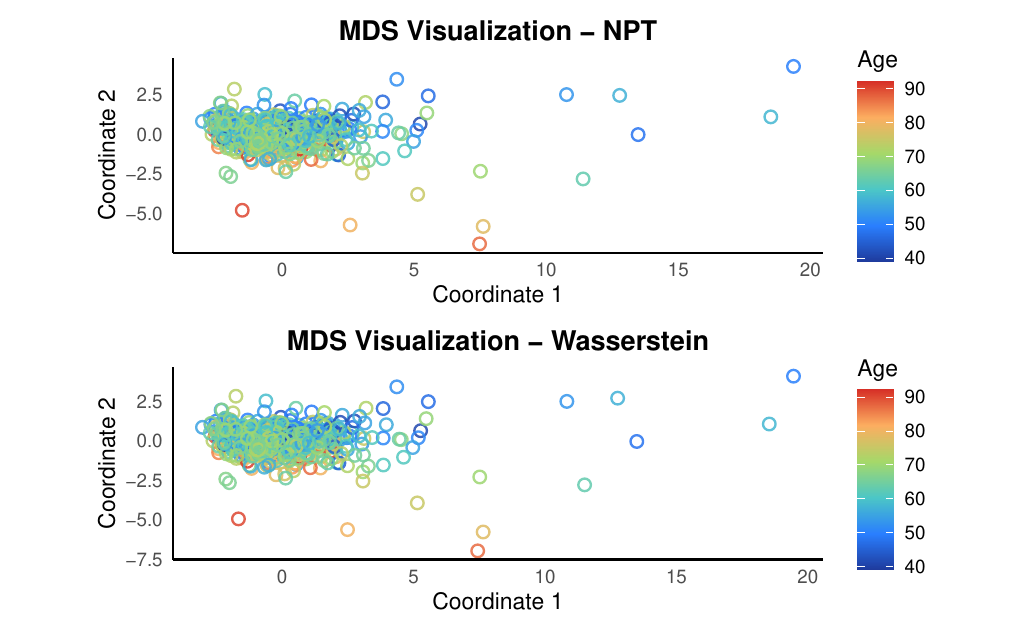}
  \caption{Scatterplots of the first two coordinates from a
    two-dimensional MDS for NPT (top) and Wasserstein (bottom) distance
    metrics on SHHS data. Hollow circles represent individual
    distributions and are color-coded by age.}
  \label{fig:mds_age}
\end{figure}

\FloatBarrier
\bibliographystyle{chicago}
\bibliography{RevisedRef}
\end{document}